\documentclass[a4paper,reqno]{amsart}
\usepackage{amssymb}
\usepackage{latexsym}
\usepackage{amsmath}
\usepackage{euscript}
\usepackage{graphics,color}
\usepackage[all]{xy}
\usepackage[margin=3cm]{geometry}
\usepackage{MnSymbol} 

\newcommand{\be}{\begin{equation}}
\newcommand{\ee}{\end{equation}}
\newcommand{\ba}{\begin{eqnarray}}
\newcommand{\ea}{\end{eqnarray}}
\newcommand{\baa}{\begin{eqnarray*}}
\newcommand{\eaa}{\end{eqnarray*}}
\newcommand{\bb}{\begin{thebibliography}}
\newcommand{\eb}{\end{thebibliography}}

\newcommand{\bi}[1]{\bibitem{#1}}
\newcommand{\lab}[1]{\label{#1}}
\newcommand{\re}[1]{(\ref{#1})}

\newcounter{my}
\newcommand{\he}%
   {\stepcounter{equation}\setcounter{my}%
   {\value{equation}}\setcounter{equation}0%
   }%
\newcommand{\she}%
   {\setcounter{equation}{\value{my}}%
    }%

\newtheorem{pr}{Proposition}

\newtheorem{theorem}{Theorem}[section]

\theoremstyle{definition}

\newtheorem{remark}[theorem]{Remark}

\numberwithin{equation}{section}


\begin{document}

\title[On the eigenvectors of the 5D number operator]{On the eigenvectors of the 5D 
discrete Fourier\\ transform number operator in Newtonian basis}

\author{Natig Atakishiyev}
\address{Universidad Nacional Aut{\'o}noma de M{\'e}xico, Instituto de Matem{\'a}ticas,  
Unidad Cuernavaca,  \newline
Cuernavaca, 62210, Morelos, M{\'e}xico}

\begin{abstract}
A simple analytic approach to the evaluation of the eigenvalues and eigenvectors 
$f_n$ of the $5D$ discrete number operator $\mathcal{N}_5=A^{\intercal}_5 A_5$ 
is formulated. This approach is based on the symmetry of the intertwining operators 
$A_5$ and $A^{\intercal}_5$ with respect to the discrete reflection operator. A 
procedure for {\it sparsealization} the intertwining operators $A_5$ and $A^{\intercal}_5$ 
has been developed, which made it possible to establish a discrete analog of the well-known 
continuous-case formula $\psi_n(x) = \frac{1}{\sqrt{n!}}\, ({\bf {a}^{\dagger}})^n \psi_0(x)$. 
A discrete analog for the eigenvectors $f_n$ of another continuous-case formula $\psi_n(x) 
= c_n^{-1} H_n(x)\,\psi_0(x),\,\,c_n= \sqrt{2^n\, n!}$, is constructed in terms of the Newtonian 
basis polynomials ${\mathcal P}_n (X_5),\,n\in {\mathbb Z_5}$, times the lowest eigenvector 
$f_0$, as well. 
\end{abstract}

\maketitle
\section{Introduction}
First, let me recall that the discrete (finite) Fourier transform (DFT) based on five 
points is represented by a $5\times 5$ unitary symmetric matrix $\Phi_5$  with 
entries (see, for example, \cite{McCPar}--\cite{Ata})
\be
(\Phi_5)_{kl} = 5^{-1/2} \: q^{kl}, \qquad k,l \in {\mathbb Z_5}:=\{0,1,2,3,4\},
\label{DFT}
\ee
where $ q=\exp{\left(2 \pi {\rm i}\, / 5\right)} $ is a primitive $5$-th root of unity. 
Those vectors $f_k$, which are solutions of the standard equations
\be
\sum_{k=0}^{4} {(\Phi_5)_{m,n}\, (f_k)_n} = {\lambda}_k\,(f_k)_m , 
\qquad k\in {\mathbb Z_5},
\label{ev}
\ee
then represent five eigenvectors of operator $\Phi_5$, associated with the eigenvalues 
${\lambda}_k$. Because the fourth power of $\Phi_5$ is a unit matrix, only four distinct 
eigenvalues among ${\lambda}_k$s are $\pm 1$ and $\pm {\rm i}$.

In addition, the discrete analog of the reflection operator $P$ (defined on the full real line 
$x\in {\mathbb R}$ as $P\,x = -\, x$), associated with the DFT operator (\ref{DFT}) is 
represented by the $5 \times 5$ matrix:
\be
P_d:= C^{\intercal}_5 J_5 \equiv J_5\,C_5\,,
\label{Pd}
\ee
where $C_5$ is the {\it 5D basic circulant permutation} matrix with entries $ (C_5)_{kl}
={\delta}_{k,l-1}$ and $J_5$ is the $5 \times 5$  {\lq}backward identity{\rq} permutation 
matrix with ones on the secondary diagonal (see \cite{HJ}, pages 26 and 28, respectively). 
It is readily verified that the DFT operator (\ref{DFT}) is $P_d$-symmetric, that is, the 
commutator $[\Phi_5, P_d] = \Phi_5\,P_d - P_d\,\Phi_5=0$. Therefore, similar to the 
continuous case, governed by the  reflection operator $P$, the eigenvectors of the DFT
operator $\Phi_5$ should be either $P_d$-symmetric or $P_d$-antisymmetric.

In the present work, additional findings concerning algebraic properties of two intertwining 
operators associated with the DFT matrix \re{DFT}  are discussed. These operators are 
represented by matrices $A_5$ and $A^{\intercal}_5$ of the same size $5\times 5$, such 
that the intertwining relations 
\be
A_5 \,\Phi_5 ={\rm i}\,\Phi_5 A_5, \qquad A^{\intercal}_5 \,\Phi_5 
= -\, {\rm i}\, \Phi_5 \, A^{\intercal}_5, 
\lab{inter} 
\ee
are valid. Matrices $A_5$ and $A^{\intercal}_5$ have emerged in a paper \cite{MesNat} 
devoted to the problem of determining an explicit form for the  difference operator that 
governs the eigenvectors of the DFT matrix $\Phi_5$. They can be interpreted as discrete 
analogs of the quantum harmonic oscillator lowering and raising operators ${\bf a}= 2^{-1/2}
\Big( x + \frac{d}{dx}\Big)$ and ${\bf {a}^{\dagger}}= 2^{-1/2}\Big( x - \frac{d}{dx}\Big)$; 
their algebraic properties have been studied in detail in \cite{AA2016}--\cite{AAZh}.  In particular, 
it was shown that the operators $A$ and $A^{\intercal}$ form a cubic algebra $\mathcal{C}_q$ 
with $q$ a root of unity \cite{AAZh}. This algebra is intimately related to the two other well-known 
realizations of the cubic algebra: the Askey-Wilson algebra \cite{Zhe_AW}--\cite{Tom2} and 
the Askey-Wilson-Heun algebra \cite{Bas_H}. This particular cubic algebra $\mathcal{C}_q$,
associated with operators $A$ and $A^{\intercal}$, is certainly more complicated than 
Heisenberg-Weyl algebra, generated by the lowering and raising operators ${\bf a}$ 
and ${\bf {a}^{\dagger}}$ of the continuous case. Nevertheless, the remarkable fact 
is that it turns out possible to use the same procedure of constructing the eigenvectors 
of the discrete number operator $\mathcal{N}:=A^{\intercal} A$ in the form of the 
ladder--type hierarchy, as in the linear harmonic oscillator case in quantum mechanics 
\cite{LL}:
\be
{\bf a}\,\psi_0(x) = 0, \qquad  \psi_n(x) = \frac{1}{\sqrt{n}}\,{\bf {a}^{\dagger}}
\,\psi_{n-1}(x), \qquad n=1,2,3,...\,.
\label{LO}
\ee
Note also that from the intertwining relations (\ref{inter}) it follows that the operator 
$\mathcal{N}_5=A^{\intercal}_5 A_5$ commutes with the DFT operator $\Phi_5$, that 
is, $[\mathcal{N}_5,\Phi_5]=0$. The discrete number operator $\mathcal{N}_5$ and DFT 
operator $\Phi_5$ thus have the same eigenvectors and the former can be employed
to find an explicit form of the eigenvectors of the latter (see \cite{MesNat} for a more 
detailed discussion of this point). 

This  paper presents a novel analytical method for evaluating the eigenvalues and 
eigenvectors of the $5D$ discrete number operator $\mathcal{N}_5=A^{\intercal}_5 A_5$,
levering the symmetry of the intertwining operators $A_5$ and $A^{\intercal}_5$ 
with respect to the discrete reflection operator $P_d$. It turned out that in order to 
achieve this goal, it is necessary to isolate from the intertwining operators $A_5$ 
and $A^{\intercal}_5$ those symmetric and antisymmetric parts that annihilate 
arbitrary vectors of the same parity. In fact, this study develops a procedure for 
the {\it sparsealization} of operators $A_5$ and $A^{\intercal}_5$, which enables 
us to derive a still-missing discrete matrix analog of the well-known formulas 
\be
\psi_n(x) = \frac{1}{\sqrt{n!}}\, ({\bf {a}^{\dagger}})^n \psi_0(x)= 
c_n^{-1} H_n(x)\,\psi_0(x) ,\qquad c_n= \sqrt{2^n\, n!}
\label{psi}
\ee
associated with the continuous case for eigenvectors $f_k$ of the DFT operator $\Phi_5$.

The remainder of this paper is organized as follows. In Section 2 an account is given 
on how to resolve the problem of finding the eigenvectors $f_n$ and eigenvalues 
$\lambda_n$ of the $5D$  discrete number operator $\mathcal{N}_5=A^{\intercal}_5 A_5$. 
A detailed description of how to accomplish this task is provided in  Section 3. In Section 4, 
an explicit form of the ladder-type hierarchy is established, in which the eigenvectors $f_n$ 
form. Section 4 closes the paper by expressing the eigenvectors $f_n$  in terms of Newtonian 
basis polynomials times the lowest eigenvector $f_0$.

\section{5D intertwining operators $A_5$ and $A^{\intercal}_5$}

This section begins by deriving additional symmetry properties of the 5D intertwining 
operators $A_5$ and $A^{\intercal}_5$ \cite{MesNat}. The explicit form of the matrices, 
associated with the operators $A_5$ and $A^{\intercal}_5$, is
\be
A_5= \frac{1}{\sqrt 2}\Big(X_5 + {\rm i} Y_5\Big) =  \frac{1}{\sqrt 2}\Big(X_5 + D_5\Big)\,, 
\qquad A^{\intercal}_5= \frac{1}{\sqrt 2}\Big(X_5 - {\rm i} Y_5\Big) =  
\frac{1}{\sqrt 2}\Big(X_5 -  D_5\Big)\,,
\label{intop}
\ee
where $X_5=diag\,(\mathsf{s}_0,\mathsf{s}_1,\mathsf{s}_2,\mathsf{s}_{3},\mathsf{s}_{4}),$\,
$\mathsf{s}_n:=2\sin(2\pi n/5),\,n \in {\mathbb Z_5},$ and $Y_5= -\,\rm{ i}\,D_5 =  \rm{ i}\, 
(C^{\intercal}_5 - C_5).$ The 5D operators $X_5$ and $Y_5$ are Hermitian and act as 
finite-dimensional analogs of the coordinates and momentum operators, respectively, in quantum 
mechanics. It is remarkable that operators $X$ and $Y$ are "classical" operators with good spectral 
properties \cite{AAZh}. For operator $X_5$, the spectrum of $X_5$ is 
\be
\lambda_n\,=\,\mathsf{s}_n={\rm i}(q^{-\,n} - q^{n}),\qquad n \in {\mathbb Z_5}.
\lab{spX5}
\ee
This indicates that the spectrum (\ref{spX5}) belongs to the class of Askey--Wilson spectra 
of the type
\be
\lambda_n\,=\,C_1 q^{n} + C_2 q^{-\,n} + C_0\,.
\lab{AWsp}
\ee
The eigenvectors of operator $X_5$ are represented by the Euclidean $5$-column 
orthonormal vectors $e_k$ with the components $(e_k)_l= \delta_{kl},\, k,l \in 
{\mathbb Z_5},$  that is,
\be
X_5\, e_k = \mathsf{s}_k\,e_k\,.
\lab{X5ev}
\ee
The spectrum of the matrix $Y_5= -\,\rm{ i}\,D_5$ belongs to the same Askey--Wilson 
family because the operators $X_5$ and $Y_5= -\,\rm{ i}\,D_5$ are unitary equivalent, 
$Y_5= -\,\rm{ i}\,D_5=\Phi_5 X_5 {\Phi}_5^{\dagger}$, and hence isospectral \cite{AAZh}. 
Note that the spectrum of $X_5$ is simple, that is, it is nondegenerate. In addition, from 
the unitary equivalence of operators $X_5$ and $Y_5= -\,\rm{ i}\,D_5$ it follows that the 
eigenvectors of the latter operator are of the form:
\be
Y_5\, {\epsilon}_k= -\,\rm{ i}\,D_5\, {\epsilon}_k = \mathsf{s}_k\,{\epsilon}_k\,, \qquad
{\epsilon}_n := \Phi_5\, e_n = 5^{-1/2} \Big(1, q^n, q^{2n}, q^{3n}, q^{4n}\Big)^{\intercal}\,. 
\lab{Y5ev}
\ee

Let me draw attention now to the remarkable symmetry between the operators 
$X_5$ and $Y_5=-\,\rm{ i}\,D_5$: the operator $X_5$  is two-diagonal in the 
eigenbasis of the operator $Y_5=-\,\rm{ i}\,D_5$, 
\be
X_5\, {\epsilon}_n = {\rm i}\,({\epsilon}_{n-1} - {\epsilon}_{n+1})\,, 
\lab{X5epsn}
\ee
whereas  operator $Y_5$  is similarly two-diagonal in the eigenbasis of 
operator $X_5$,  
\be
Y_5\, e_n = -\,\rm{ i}\,D_5 \,e_n= {\rm i}\,(e_{n+1} - e_{n-1})\,.
\lab{Y5en}
\ee
It is also worth mentioning here that the $N$-column eigenvectors of 
operator $Y = -\,\rm{ i}\,D$ for a general $N$,
\be
{\epsilon}_n = \Phi_N\, e_n = \sum_{k=0}^{N-1} {(\Phi_N)_{kn}\, e_k} 
= N^{-1/2} \Big(1, q^n, q^{2n},\dots, q^{(N-1)n}\Big)^{\intercal}\,, 
\lab{d_Phi_e}
\ee
form an orthonormal basis in the $N$-dimensional complex plane ${\mathbb C}^N$ and 
are frequently used therefore as building blocks of the discrete Fourier transform in 
applications (see, for example, p.130 in \cite{HarmAnal}, where the ${\epsilon}_n$
referred to as {\it discrete trigonometric functions}).

Note also that both operators  $X_5$ and $\,D_5 =\,\rm{ i}Y_5$ are 
$P_d$-antisymmetric; that is,
\be
P_d\, X_5 +  X_5 P_d = 0, \qquad P_d \,D_5  + D_5\, P_d =0\,.
\lab{X5D5}
\ee
Moreover, since the 5D intertwining operators $A_5$ and $A^{\intercal}_5$ under
discussion are linear combinations of the operators  $X_5$ and $\,D_5 =\,\rm{ i}Y_5$,
from (\ref{X5D5}) it follows that both operators $A_5$ and $A^{\intercal}_5$  are 
also $P_d$-antisymmetric. 

It remains only to note that such a detailed discussion of the matrix structure of the 
operators  $X_5$ and $\,D_5 =\,\rm{ i}Y_5$ is dictated by the fact that it makes it 
possible to significantly simplify the problem of finding the eigenvectors of the 5D 
discrete number operator $\mathcal{N}_5 = A^{\intercal}_5\, A_5$. It is possible 
to formulate a simpler algorithm for calculating the eigenvectors of the discrete 
number operator $\mathcal{N}_5 = A^{\intercal}_5\, A_5$ by separating from the 
operators $A_5$ and $A^{\intercal}_5$ their symmetric and antisymmetric parts, 
which are essentially annihilators. A detailed description of how to achieve this goal 
is provided in the next section, where it becomes apparent that the key idea of this 
approach is the essential use of the remarkably symmetric matrix structure of the
product $\Phi_5 X_5$. Therefore, let me close this section with a discussion about 
a particular $P_d$-symmetry property of the product $\Phi_5 X_5$, which will be 
essentially used in what follows.
\begin{pr}
The product $\Phi_5 X_5$ can be represented as either
\be
\Phi_5 X_5 = s^{-1}_2\,{\mathcal A}^{(s)} + {\rm i}\, {\mathcal B}^{(s)},
\lab{Ipart}
\ee
where ${\mathcal A}^{(s)}$ is a symmetric annihilator operator that annuls 
every $P_d$-symmetric vector $f^{(s)}:=(a,b,c,c,b)^{\intercal}$, 
and ${\mathcal B}^{(s)}$ is a sparse matrix, or 
\be
\Phi_5 X_5 = s^{-1}_2\,\Big({\mathcal A}^{(a)}\,+ \,{\mathcal B}^{(a)}\Big),
\lab{IIpart}
\ee
where ${\mathcal A}^{(a)}$ is an antisymmetric annihilator operator that annuls 
every $P_d$-antisymmetric vector $f^{(a)}:=(0,b,c,- c,- b)^{\intercal}$, and 
${\mathcal B}^{(a)}$ is a sparse matrix. 
\end{pr}
\begin{proof}
The operator $\Phi_5 X_5$ is represented by a traceless matrix
\be
\Phi_5 X_5 = \frac{1}{s_2}\,\left[\begin{array}{ccccc} 
0   &    1     &      c_1      &       - c_1     &  -1       \\
0   &    q     &  c_1 q^2  &  - c_1 q^3  &  - q^4  \\
0   &  q^2  &  c_1 q^4  &    - c_1 q     &  - q^3  \\
0   &  q^3  &   c_1 q      & - c_1 q^4   &  - q^2  \\
0   &  q^4  &   c_1 q^3 & - c_1 q^2   &  -  q
\end{array}\right]\,.
\lab{phi5x5}
\ee
From (\ref{phi5x5}) it follows that the first row annihilates an arbitrary 
$P_d$-symmetric vector $f^{(s)}:=(a,b,c,c,b)^{\intercal}$, that is, 
$(0,1,c_1, - c_1, -1)\,f^{(s)} = 0.$  In addition, with the aid of simple 
identities $q = q^4 + {\rm i} s_1$ and $q^2 = q^3 + {\rm i} s_2$, the 
second and third rows in  (\ref{phi5x5}) can be rewritten as 
\[
\Big(0, q^4 + {\rm i} s_1, c_1(q^3 + {\rm i} s_2), - c_1 q^3,- q^4\Big)=
\Big(0, q^4, c_1 q^3, - c_1 q^3,- q^4\Big) + {\rm i}\,\Big(0, s_1, c_1 s_2, 0, 0\Big), 
\]
\be
\Big(0, q^3 + {\rm i} s_2, c_1(q - {\rm i} s_1), - c_1 q,- q^3\Big)= 
\Big(0, q^3, c_1 q, - c_1 q,- q^3\Big) + {\rm i}\,\Big(0, s_2, -  s_2, 0, 0\Big), 
\lab{2r3r}
\ee
respectively.

Similarly, the fourth and fifth rows in (\ref{phi5x5}) can be represented
as  
\[
\Big(0, q^3, c_1 q,  - c_1 (q - {\rm i} s_1) ,- (q^3 + {\rm i} s_2)\Big)=
\Big(0, q^3, c_1 q, - c_1 q,- q^3\Big) + {\rm i}\,\Big(0, 0, 0, s_2, - s_2\Big), 
\]
\be
\Big(0, q^4 , c_1 q^3, - c_1( q^3 + {\rm i} s_2), - ( q^4 + {\rm i} s_1)\Big)= 
\Big(0, q^4, c_1 q^3, - c_1 q^3,- q^4\Big) -  {\rm i}\,\Big(0, 0, 0, c_1 s_2, s_1\Big), 
\lab{2r3r}
\ee
respectively. Hence the initial matrix (\ref{phi5x5}) can be divided into two parts,
$\Phi_5 X_5 = s^{-1}_2\,{\mathcal A}^{(s)} + {\rm i}\, {\mathcal B}^{(s)}$,
where
\be
{\mathcal A}^{(s)} = \left[\begin{array}{ccccc} 
0   &    1     &      c_1      &    - c_1     &  -1       \\
0   &  q^4  &  c_1 q^3  & - c_1 q^3 &  - q^4  \\
0   &  q^3  &   c_1 q      &   - c_1 q   &  - q^3  \\
0   &  q^3  &   c_1 q      &   - c_1 q   &  - q^3  \\
0   &  q^4  &   c_1 q^3 & - c_1 q^3 & - q^4
\end{array}\right]\,,\qquad
{\mathcal B}^{(s)} = \left[\begin{array}{ccccc} 
0   &    0    &      0    &   0    &     0   \\
0   & - c_2 &   c_1   &   0    &     0    \\
0   &    1    &    - 1   &   0     &    0    \\
0   &    0    &     0    &   1     &  - 1    \\
0   &    0    &     0    & - c_1 &  c_2
\end{array}\right]\,.
\lab{macABs}
\ee
From the explicit form of the matrix ${\mathcal A}^{(s)}$ in (\ref{macABs}) 
it is evident that ${\mathcal A}^{(s)} f^{(s)}=0$. In addition, with only  
eight non-zero elements,  the $5\times 5$ matrix  ${\mathcal B}^{(s)}$ is 
a sparse matrix. Consequently, identity (\ref{Ipart}) is proved with the 
identification of the explicit forms of both operators ${\mathcal A}^{(s)}$ 
and the matrix ${\mathcal B}^{(s)}$.

Alternatively, the same matrix (\ref{phi5x5}) can be split into two parts: $\Phi_5 X_5 
= s^{-1}_2\,\Big({\mathcal A}^{(a)}\, + \, {\mathcal B}^{(a)}\Big)$, where 
\be
{\mathcal A}^{(a)} = \left[\begin{array}{ccccc} 
0   &  - 1      &   -  c_1       &     - c_1     &  -1       \\
0   & - q^4  & -  c_1 q^3 & - c_1 q^3 &  - q^4  \\
0   & - q^3  & -  c_1 q      &  - c_1 q    &  - q^3  \\
0   &  q^3   &    c_1 q      &    c_1 q     &    q^3  \\
0   &  q^4   &   c_1 q^3  &   c_1 q^3  &   q^4
\end{array}\right]\,,\qquad
{\mathcal B}^{(a)} = \left[\begin{array}{ccccc} 
0   &     2   &   2 c_1  &       0     &     0   \\
0   &  c_1  &   - 1      &        0     &     0    \\
0   &  c_2  &  c^2_1 &        0     &     0    \\
0   &    0    &     0      & - c^2_1 &  - c_2 \\
0   &    0    &     0      &     1        &   - c_1
\end{array}\right]\,.
\lab{macABa}
\ee
From the explicit form of matrix ${\mathcal A}^{(a)}$ in (\ref{macABa}) it 
is evident that ${\mathcal A}^{(a)} f^{(a)}=0$. The $5\times 5$ matrix  
${\mathcal B}^{(a)}$, with only 10 nonzero elements,  is also a sparse 
matrix. Consequently, identity (\ref{IIpart}) is also proved, with the 
identification of the explicit forms of both operators ${\mathcal A}^{(a)}$ 
and the matrix ${\mathcal B}^{(a)}$. This completes the proof of the proposition.
\end{proof}

\section{Eigenvectors and eigenvalues of the discrete number operator $\mathcal{N}_5$}

This section begins with the derivation of an explicit form of the eigenvectors and 
eigenvalues of discrete number operator $\mathcal{N}_5$.

{\bf 1°}. Similar to the continuous  case ({\ref{LO}), the lowest eigenvector $f_0$ of 
the discrete number operator $\mathcal{N}_5$ is obtained by solving the difference
equation 
\be
A_5\, f_0\,=\,\frac{1}{\sqrt 2} \Big(X_5 + D_5\Big)\,f_0 = 0\,.
\lab{f0}
\ee
This vector $f_0$ must be $P_d$-symmetric, i.e. it can be written as 
$f_0=(x_0,x_1,x_2,x_2,x_1)^{\intercal}.$ Therefore the matrix form 
of  Equation (\ref{f0}) is: 
\[
A_5\,f_0=\frac{1}{\sqrt 2} \left[\begin{array}{ccccc} 
0   &   1    &  0   &   0    &  -1  \\
-1  &  s_1 &  1   &   0    &   0  \\
0   &  -1    & s_2&   1    &   0 \\
0   &   0    & -1  & - s_2 &   1 \\
1   &   0    &  0  & - 1     & - s_1
\end{array}\right]
\left[\begin{array}{c}
x_0    \\
x_1   \\
x_2  \\
x_2  \\
x_1 
\end{array}\right]=
\frac{1}{\sqrt 2} \left[\begin{array}{c}
                   0              \\
s_1 x_1 + x_2  - x_0  \\
s_2 x_2  + x_2  - x_1 \\
x_1 -  s_2 x_2  -  x_2 \\
x_0  -  s_1 x_1  -x_2
\end{array}\right]= 0\,.
\lab{Maf0}
\] 
Thus, only two linearly independent equations 
\be
x_0\,=\,s_1 x_1 + x_2  \,, \quad  x_1\,=\,x_2\,(1 + s_2)\,,
\lab{cof0}
\ee
interconnecting components $x_0$, $x_1$ and $x_2$  were obtained.
Hence the lowest eigenvector 
\be
f_0\,=\,x_2\, ( \xi_0, \xi_1, 1, 1, \xi_1)^{\intercal}, \qquad  
\xi_0\,=\,s_1 - 2 c_2, \quad \xi_1\,=\,1 + s_2\,,
\lab{f0ex}
\ee
is determined up to the multiplicative factor $x_2$, the explicit form of which 
is given as follows.

It is directly verified that $\Phi_5 f_0 = f_0$; moreover, from this formula 
and the second identity in \re{inter} it follows that: 
\be
\Phi_5 f_k\,=\,{\rm{i}}^k f_k\,,\qquad k\in {\mathbb Z_5}\,.
\lab{evfkPhi}
\ee
Recall that  the functions $\psi_n(x)$ from  \re{psi} possess a simple 
transformation property with respect to the Fourier transform: {\it they are 
eigenfunctions of the Fourier transform, associated with the eigenvalues 
${\rm i}^{\,n}$},
\be
\left({\bf{\mathcal F}}\,\psi_n\right)(x) \equiv \frac{1}{\sqrt{2\pi}}\,
\int_{\mathbb R}\,e^{\,{\rm i}\,x\,y}\,\psi_n(y)\,dy = {\rm i}^{\,n}
\,\psi_n(x)\,.                                                                         
\lab{FTpsi}
\ee
Thus, Equation \re{evfkPhi} is a discrete analog of the continuous case  \re{FTpsi}.

It should be emphasized that owing to the simplicity of defining Equation ({\ref{f0}) for 
the lowest eigenvector $f_0$, there was no need to use Proposition $1$ in the first step. 

 {\bf 2°}. To define next to the $f_0$ eigenvector $f_1$, we evaluate 
\[
A^{\intercal}_5\, f_0\,=\,\frac{1}{\sqrt 2} \Big( X_5 - D_5\Big)\,f_0 = 
\frac{1}{\sqrt 2} \Big( X_5 -\,\rm{ i}\,\Phi_5 X_5 {\Phi}_5^{\dagger} \Big)\,f_0 = 
\]
\be
= \frac{1}{\sqrt 2} \Big(X_5 -\,\rm{ i}\,\Phi_5 X_5\Big)\,f_0 = 
\frac{1}{\sqrt 2} \Big( X_5 + {\mathcal B}^{(s)} \Big)\,f_0 \,,
\lab{A(int)f0}
\ee
by using the unitary equivalence of operators $ -\,\rm{ i}\,D_5$ and $X_5$ in 
the first step, the identity ${\Phi}_5^{\dagger}\,f_0 = f_0$, which follows from 
\re{evfkPhi}, in the second step, and the identity \re{Ipart} in the third step.
Upon rewriting \re{A(int)f0} in matrix form, we arrive at the following:
\be
A^{\intercal}_5\, f_0\,=\,\frac{x_2}{\sqrt 2} \left[\begin{array}{ccccc}
    0   &        0         &       0       &      0      &    0      \\
    0   & - c_2 \xi_1 &    c_1      &      0      &    0      \\
    0   &        1         & \xi_1-2  &      0     &    0      \\
    0   &        0         &       0      & 2 - \xi_1 & - 1      \\
    0   &        0         &       0      &    - c_1   & c_2 \xi_1
\end{array}\right]  \left[\begin{array}{c}
 \xi_0  \\
 \xi_1  \\
   1     \\
   1     \\
 \xi_1
\end{array}\right]= 2^{1/2} x_2 s_1 \left[\begin{array}{c}
    0     \\
 \xi_1  \\
   c_1  \\
 -  c_1 \\
 - \xi_1
\end{array}\right]\,,
\lab{A(int)f0F}
\ee
where the readily verified identity $c_2\,\xi_1^2 = c_1  - 2 s_1 {\xi_1}$ is 
considered. Hence, the unit-length $P_d$-antisymmetric eigenvector $f_1$ 
is of the form
\be
f_1\,=\,\frac{1}{2\sqrt{s_2 \xi_1}}\,\Big( 0,  \xi_1, c_1, - c_1, -  \xi_1\Big)^{\intercal}\,.
\lab{f1}
\ee
To find explicitly the eigenvalue $\lambda_1$ of the discrete number operator 
$\mathcal{N}_5 = A^{\intercal}_5\, A_5$, associated with the eigenvector $f_1$,
one evaluates first
\[
A_5\, f_1\,=\,\frac{1}{\sqrt 2} \Big( X_5 + D_5\Big)\,f_1 = 
\frac{1}{\sqrt 2} \Big( X_5 +\,\rm{ i}\,\Phi_5 X_5 {\Phi}_5^{\dagger} \Big)\,f_1 = 
\]
\[
=\frac{1}{\sqrt 2} \Big(X_5 + \Phi_5 X_5\Big)\,f_1 = 
\frac{1}{\sqrt 2} \Big( X_5 + s^{-1}_2\,{\mathcal B}^{(a)} \Big)\,f_1 =  u\,,
\]
\be
u\,=\,\frac{1}{2\sqrt{2 s_2 \xi_1}}\,\Big( 2 \xi_1,  s_1 \xi_1 +  c_1, 
c_2 - c^2_1 s_2, c_2 - c^2_1 s_2, s_1 \xi_1 +  c_1,\Big)^{\intercal}\,,
\lab{A5f1}
\ee
by using the unitary equivalence of operators $ -\,\rm{ i}\,D_5$ and $X_5$ in the 
first step, the identity ${\Phi}_5^{\dagger}\,f_1 = - {\rm i} f_1$, which follows from 
\re{evfkPhi}, in the second step, and the identity \re{IIpart} in the third step. It is 
readily verified that the $P_d$-symmetric vector $u$ is the eigenvector of the DFT 
operator $\Phi_5$, that is, $\Phi_5\,u = u$. Therefore one may use identity \re{Ipart} 
to evaluate
\[
A^{\intercal}_5 A_5\, f_1\,=\,A^{\intercal}_5\,u = 
\frac{1}{\sqrt 2} \Big( X_5 - D_5\Big)\,u = 
\frac{1}{\sqrt 2} \Big( X_5 -\,\rm{ i}\,\Phi_5 X_5 {\Phi}_5^{\dagger} \Big)\,u = 
\]
\be
=\frac{1}{\sqrt 2} \Big(X_5 -\,\rm{ i} \Phi_5 X_5\Big)\,u = 
\frac{1}{\sqrt 2}  \Big( X_5 +\,{\mathcal B}^{(s)} \Big)\,u\,=
\,\lambda_1\,f_1\,,\qquad \lambda_1= [c_1 ( s_2 - 1) + 7]/2 \,.
\lab{lambda1}
\ee

 {\bf 3°}. To define the next eigenvector $f_2$, we evaluate 
\[
A^{\intercal}_5\, f_1\,=\,\frac{1}{\sqrt 2} \Big( X_5 - D_5\Big)\,f_1 = 
\frac{1}{\sqrt 2} \Big( X_5 -\,\rm{ i}\,\Phi_5 X_5 {\Phi}_5^{\dagger} \Big)\,f_1 = 
\]
\be
= \frac{1}{\sqrt 2} \Big(X_5 - \Phi_5 X_5\Big)\,f_1 
= \frac{1}{\sqrt 2} \Big( X_5 - s^{-1}_2 {\mathcal B}^{(a)} \Big)\,f_1 \,,
\lab{A(int)f1}
\ee
using  the identity ${\Phi}_5^{\dagger}\,f_1 = - \rm{ i}\, f_1$ and identity \re{IIpart}. 
Upon rewriting the right-hand side of \re{A(int)f1} in matrix form, we obtain 
\be
A^{\intercal}_5\, f_1\,=\,\frac{1}{2{\xi}^{1/2}_1\,{(2 s)}^{3/2}_2}
    \left[\begin{array}{ccccc}
    0   &      - 2       &   - 2 c_1  &      0     &    0      \\
    0   &    - c_2     &      1        &      0     &    0      \\
    0   &    - c_2     &      1        &      0     &    0      \\
    0   &        0       &       0      &     - 1    &   c_2    \\
    0   &        0       &       0      &    -  1    &   c_2 
\end{array}\right]  \left[\begin{array}{c}
    0     \\
 \xi_1  \\
   c_1  \\
 - c_1  \\
 - \xi_1
\end{array}\right]= \sqrt{s_1 (s_1 - c_2)/2}\,f_2\,,
\lab{A(int)f1F}
\ee
where the unit-length $P_d$-symmetric eigenvector $f_2$ is of the form
\be
f_2\,=\,\frac{1}{2\,s_2}\,\Big( -  2 c_1,  1, 1, 1, 1\Big)^{\intercal}\,.
\lab{f2}
\ee
To find explicitly the eigenvalue $\lambda_2$ of the discrete number operator 
$\mathcal{N}_5 = A^{\intercal}_5\, A_5$, associated with the eigenvector $f_2$,
one evaluates first
\[
A_5\, f_2\,=\,\frac{1}{\sqrt 2}\Big( X_5 + D_5\Big)\,f_2 
= \frac{1}{\sqrt 2}\Big( X_5 +\,\rm{ i}\,\Phi_5 X_5 {\Phi}_5^{\dagger} \Big)\,f_2 = 
\]
\be
= \frac{1}{\sqrt 2} \Big(X_5 -\,\rm{ i} \Phi_5 X_5\Big)\,f_2 = 
\frac{1}{\sqrt 2}\Big( X_5 + \,{\mathcal B}^{(s)} \Big)\,f_2 = \sqrt{s_1 (s_1 - c_2)/2}\,\,f_1\,,
\lab{A5f2}
\ee
using the identity ${\Phi}_5^{\dagger}\,f_2 = - f_2$, which follows from \re{evfkPhi}, 
and identity \re{Ipart}. Then, from \re{A(int)f1F} and \re{A5f2} it follows that 
\be
A^{\intercal}_5 A_5\, f_2\,=\,\sqrt{s_1 (s_1 - c_2)/2} \, A^{\intercal}_5\,f_1 = 
[s_1 (s_1 - c_2)/2]\,f_2\,=\,\lambda_2\,f_2\,,\qquad \lambda_2= s_1 (s_1 - c_2)/2\,.
\lab{lambda2}
\ee

 {\bf 4°}. To define the next eigenvector $f_3$, we evaluate
\[
A^{\intercal}_5\, f_2\,=\, \frac{1}{\sqrt 2} \Big( X_5 - D_5\Big)\,f_2 = 
 \frac{1}{\sqrt 2}\Big( X_5 -\,\rm{ i}\,\Phi_5 X_5 {\Phi}_5^{\dagger} \Big)\,f_2 = 
\]
\be
=  \frac{1}{\sqrt 2} \Big(X_5 +\,\rm{ i}\,\Phi_5 X_5\Big)\,f_2 = 
 \frac{1}{\sqrt 2}\Big( X_5 -  {\mathcal B}^{(s)} \Big)\,f_2 \,,
\lab{A(int)f2}
\ee
using the identity ${\Phi}_5^{\dagger}\,f_2 = - f_2$, which follows from 
\re{evfkPhi},  and identity \re{Ipart}. Upon rewriting the right-hand side
of Equation \re{A(int)f2} in matrix form, we arrive at the following:
\be
A^{\intercal}_5\, f_2\,=\,\frac{1}{2^{3/2}  s_2} \left[\begin{array}{ccccc}
    0   &        0         &       0       &      0      &    0      \\
    0   & s_1 + c_2  &   -  c_1    &      0      &    0      \\
    0   &      - 1         &   \xi_1    &      0      &     0      \\
    0   &        0         &       0      & - \xi_1   &    1      \\
    0   &        0         &       0      &   c_1     & - (s_1 + c_2)
\end{array}\right]  \left[\begin{array}{c}
 - 2 c_1 \\
     1     \\
    1     \\
    1     \\
    1
\end{array}\right]= \frac{1}{2^{3/2} c_1} \left[\begin{array}{c}
      0     \\
 1 - s_2 \\
    c_1   \\
 -  c_1  \\
 s_2 - 1
\end{array}\right]\,.
\lab{A(int)f2F}
\ee
Hence the unit-length $P_d$-antisymmetric eigenvector $f_3$ is of the form
\be
f_3\,=\,\frac{1}{2\sqrt{s_2 (s_2 - 1)}}\,\Big( 0, 1 -  s_2, c_1, -\,c_1, s_2 - 1\Big)^{\intercal}\,.
\lab{f3}
\ee
To find explicitly the eigenvalue $\lambda_3$ of the discrete number operator 
$\mathcal{N}_5 = A^{\intercal}_5\, A_5$, associated with the eigenvector $f_3$,
one evaluates first
\[
A_5\, f_3\,=\,\frac{1}{\sqrt 2} \Big( X_5 + D_5\Big)\,f_3 = 
\frac{1}{\sqrt 2} \Big( X_5 +\,\rm{ i}\,\Phi_5 X_5 {\Phi}_5^{\dagger} \Big)\,f_3 = 
\]
\be
=\frac{1}{\sqrt 2} \Big(X_5 - \Phi_5 X_5\Big)\,f_3 = 
\frac{1}{\sqrt 2}\Big( X_5 -  s^{-1}_2\,{\mathcal B}^{(a)} \Big)\,f_3 = \sqrt{s_1 (s_1 + c_2)/2}\,f_2\,,
\lab{A5f3}
\ee
using the identity ${\Phi}_5^{\dagger}\,f_3 =  {\rm  i} f_3$, which follows from 
\re{evfkPhi},  and identity \re{IIpart}. Then, from \re{A(int)f2F} and \re{A5f3}, 
it follows that 
\be
A^{\intercal}_5 A_5\, f_3\,=\,\sqrt{s_1 (s_1 + c_2)/2}\,A^{\intercal}_5\,f_2 = 
[s_1 (s_1 + c_2)/2]\,f_3 = \lambda_3\,f_3\,,\qquad \lambda_3 = s_1 (s_1 + c_2)/2\,.
\lab{lambda3}
\ee

{\bf 5°}. Finally, to define the last eigenvector $f_4$, we evaluate
\[
A^{\intercal}_5\, f_3\,=\, \frac{1}{\sqrt 2} \Big( X_5 - D_5\Big)\,f_3 = 
\frac{1}{\sqrt 2}\Big( X_5 -\,\rm{ i}\,\Phi_5 X_5 {\Phi}_5^{\dagger} \Big)\,f_3 = 
\]
\be
=\frac{1}{\sqrt 2} \Big(X_5 + \Phi_5 X_5\Big)\,f_3 = 
\frac{1}{\sqrt 2}\Big( X_5  +  s^{-1}_2 {\mathcal B}^{(a)} \Big)\,f_3 \,,
\lab{A(int)f3}
\ee
using  the identity ${\Phi}_5^{\dagger}\,f_3 = \rm{ i}\, f_3$ and identity \re{IIpart}. 
Upon rewriting the right-hand side of \re{A(int)f3} in matrix form, we obtain 
\be
A^{\intercal}_5\, f_3\,=\,\frac{1}{2 c_1 s_2 \sqrt{2 \lambda_3}}
    \left[\begin{array}{ccccc}
    0   &         2        &      2 c_1   &       0      &         0       \\
    0   &  3 c_1 + 1  &     -  1       &       0      &         0       \\
    0   &       c_2      &  3 - 2 c_ 1 &       0      &         0       \\
    0   &        0        &        0        & 2 c_1 - 3 &    - \, c_2  \\
    0   &        0        &        0        &    -  1       & -\,(3 c_ 1+ 1) 
\end{array}\right]  \left[\begin{array}{c}
    0      \\
1 - s_2 \\
  c_1    \\
- c_1   \\
s_2 - 1
\end{array}\right]= \,v\,,
\lab{A(int)f3F}
\ee
where the $P_d$-symmetric vector $v$ is of the form
\be
v\,=\,\frac{\sqrt{\lambda_3}}{2 s_2}\,\Big(2 c_1, -\,(2s_2 + 1),  2s_2 + 3 - 2c_1,  
2s_2 + 3 - 2c_1,  - \,(2s_2 + 1)\Big)^{\intercal}\,.
\lab{v}
\ee
It is readily verified now that the $P_d$-symmetric vector $v$ is the eigenvector of 
the DFT operator $\Phi_5$, that is, $\Phi_5\,v =  v$. To find explicitly the eigenvalue 
$\lambda_4$ of the discrete number operator $\mathcal{N}_5 = A^{\intercal}_5\, A_5$, 
associated with the eigenvector $v$, one evaluates first
\[
A_5\, v\,=\,\frac{1}{\sqrt 2} \Big( X_5 + D_5\Big)\,v = 
\frac{1}{\sqrt 2} \Big( X_5 +\,\rm{ i}\,\Phi_5 X_5 {\Phi}_5^{\dagger} \Big)\,v = 
\]
\be
=  \frac{1}{\sqrt 2} \Big(X_5  + \,\rm{ i} \Phi_5 X_5\Big)\,v = 
\frac{1}{\sqrt 2} \Big( X_5 - \,{\mathcal B}^{(s)} \Big)\,v =  
\frac{1}{\sqrt 2}\,( 7 - c_1  - c_1  s_2 )\,f_3\,.
\lab{A5g2}
\ee
Then from \re{A(int)f3F} and \re{A5g2} it follows at once that 
\be
A^{\intercal}_5 A_5\, v\,=\, \frac{1}{\sqrt 2} ( 7 - c_1 -  c_1  s_2 )\, A^{\intercal}_5\,f_3 = 
\frac{1}{\sqrt 2} ( 7 - c_1 -  c_1  s_2 )\,v\,=\,\lambda_4\,v\,,
\quad \lambda_4= [7 - c_1 (1  + s_2)]/2\,.
\lab{lambda4}
\ee
Therefore, a $P_d$-symmetric eigenvector $f_4$ of unit length is simply a properly
normalized version of vector $v$, that is, $v = {\sqrt{\lambda_4}}\,f_4$ and
\be
f_4\,=\,\frac{1}{\sqrt{\lambda_2\,\lambda_4}}\,\Big( 2,  c_2 -  2 s_1,  
2 s_1 -  c_2 + 2 c_1, 2 s_1 -  c_2 +  2 c_1,  c_2 - 2  s_ 1\Big)^{\intercal}\,.
\lab{f4}
\ee
\begin{remark} 
It should be noted that by defining all the eigenvalues $\lambda_1=[c_1 (s_2 - 1)+ 7]/2 $,
$\lambda_2=s_1(s_1 - c_2)/2$, $\lambda_3=s_1(s_1 + c_2)/2$, and $\lambda_4=[7 - c_1 (1  + s_2)]/2$, 
it is possible to uniformly write all the preceding $f_4$ eigenvectors of unit length as
\[
f_0 = \frac{2}{\sqrt{\lambda_2 \lambda_4}}\,( s_1- 2 c_2, 1 + s_2, 1, 1, 1 + s_2)^{\intercal},
\qquad
f_1\,=\,\frac{1}{\sqrt{2 \lambda_2}}\,\Big( 0,  s_1- c_2, 1, - 1, c_2 - s_1 \Big)^{\intercal},
\]
\be
f_2\,=\,\frac{1}{\,\sqrt{\lambda_2\,\lambda_3}}\,\Big( 2, c_2, c_2, c_2, c_2 \Big)^{\intercal},
\qquad
f_3\,=\,\frac{1}{\sqrt{2 \lambda_3}}\,\Big( 0, - (s_1 + c_2), 1, - 1,  s_1 + c_2 \Big)^{\intercal}\,.
\lab{f0-3}
\ee
Note also that the eigenvectors $f_k,\, k \in {\mathbb Z_5}$, from  \re{f4} and \re{f0-3} are 
orthonormalized, that is,
\be
(f_k, f_l)\,=\,\delta_{k,l}, \qquad  k,l \in {\mathbb Z_5}\,.
\label{ortno}
\ee
\end{remark} 

\section{Explicit forms of discrete analogs}

{\bf 1°}. Thus, the above formulated approach to finding the eigenvectors of the 
discrete number operator $\mathcal{N}_5 = A^{\intercal}_5\, A_5$ really leads to 
the definition of a certain hierarchy of the ladder type that these eigenvectors form.
It only remains to consider that the compact form of this hierarchy still contains another 
parameter, which can be interpreted as follows.

It is well known that if $A$ and $B$ are $n\times n$ matrices, then $AB$ and $BA$ have the 
same eigenvalues (see \cite{JHW}, p.54). Hence the  operator $\mathcal{N}^{\,(s)}_5 := 
A_5\, A^{\intercal}_5$ has the same $5$ eigenvalues $\lambda_n$ as the DFT operator
$\mathcal{N}_5 = A^{\intercal}_5\, A_5$, that is,
\be
\mathcal{N}^{\,(s)}_5 g_n = A_5\, A^{\intercal}_5 \,g_n = \lambda_n \,g_n\,,
\qquad n\in {\mathbb Z_5}\,.
\lab{N{s}}
\ee
Moreover, it is  readily verified that the eigenvectors $g_n$ and $f_n$ 
of the operators $\mathcal{N}^{\,(s)}_5 = A_5\, A^{\intercal}_5$ and 
$\mathcal{N}_5 = A^{\intercal}_5\, A_5$, respectively, are interrelated
as 
\[
g_0 = \sin{\varphi}\,f_0 + \cos{\varphi}\,f_4 = \eta\,\Big( \frac{s^{2}_1}{4}\, f_0 + f_4 \Big)\,, 
\]
\[
g_1 = \cos{\varphi}\,f_0 -  \sin{\varphi}\,f_4 = \eta\,\Big( f_0  -  \frac{s^{2}_1}{4} f_4 \Big)\,,
\]
\[
g_2\,=\,f_1\,,\qquad g_3\,=\,f_2\,,\qquad g_4\,=\,f_3\,,
\]
\[
\cos{\varphi} = 4 / \sqrt {21 - 5 c_2} = \eta\,, \qquad   \sin{\varphi} = 
{s^{2}_1}/ \sqrt {21 - 5 c_2} = \eta\,{s^{2}_1}/4 \,,
\]
\be
\varphi = \arctan\Big(s^{2}_1/4\Big) = \arctan{[(5 + \sqrt{5})/8]}\,=\,42,13^{\circ}\,,
\lab{gn}
\ee
where the parameter $\eta = \cos{\varphi}$ can also be expressed  in terms of eigenvalues 
$\lambda_1$ and $\lambda_4$ as $\eta = 8 s_2/ \sqrt {\lambda_1\,\lambda_4}$.

{\bf 2°}. Having defined the parameter $\eta$, the simple geometric interpretation 
of which is obvious from (\ref{N{s}}) and (\ref{gn}), it is now easy to show that this 
parameter also enters the discrete analog of (\ref{psi}}) as  
\be
f_n \,= \,\Big( \eta \prod^n_{k=1} \lambda^{1/2}_k\Big)^{-1} 
\Big( A^{\intercal}_5\Big)^n f_0\,,\qquad n=1,2,3,4 \,.
\lab{fnf0}
\ee
Recall that in the continuous case  $ N \psi_n = {\bf {a}^{\dagger}}{\bf a}\,
\psi_n = n\,\psi_n$, hence  multiplier $n!$  in Eq. (\ref{psi}) can be expressed as 
$n! = 1 \cdot 2 \cdot \cdot \cdot n = \lambda_1 \lambda_2 \cdot \cdot \cdot \lambda_n$
with $\lambda_m:= m$, confirming the similarity between  (\ref{psi}) and  (\ref{fnf0}).

Moreover, it turns out that the parameter $\eta$ essentially contributes into 
a discrete analogue of the three-term recurrence relation \cite{LL}
\be
\sqrt{2\,(n+1)}\,\psi_{n+1}(x) + \sqrt{2 n}\,\psi_{n-1}(x)= 2x\,\psi_n(x) \,,
\qquad n = 1,2,...\,,
\label{CRR}
\ee
associated with continuous cases (\ref{psi}). This  can be expressed as follows.

From (\ref{A(int)f3F}), (\ref{v}) and (\ref{f4}) it follows that $ A^{\intercal}_5\,f_3 
= \sqrt{\lambda_4}\, f_4$. Consequently, this can be rewritten as:
\be
\sqrt{\lambda_4}\, f_4 = A^{\intercal}_5\,f_3 = ({\sqrt 2}\,X_5 - A_5)\,f_3 
= {\sqrt 2}\, X_5 f_3 - \sqrt{\lambda_3}\, f_2\,,
\label{a}
\ee
where the evident identity $ A_5 + A^{\intercal}_5 = {\sqrt 2}\, X_5$  has been used 
first, followed by the identitity  $A_5\,f_3 = \sqrt{\lambda_3}\, f_2$ from (\ref{A5f3}) 
and (\ref{lambda3}).

Similarly, from (\ref{A(int)f2F}) and (\ref{f3}) it follows that $ A^{\intercal}_5\,f_2 
= \sqrt{\lambda_3}\, f_3$. Therefore,
\be
\sqrt{\lambda_3}\, f_3 = A^{\intercal}_5\,f_2 = ({\sqrt 2}\,X_5 - A_5)\,f_2 
= {\sqrt 2}\,X_5 f_2 - \sqrt{\lambda_2}\, f_1\,,
\label{b}
\ee
because $A_5\,f_2 = \sqrt{\lambda_2}\,f_1$  by Eq.(\ref{A5f2}).
Thus both identities (\ref{a}) and (\ref{b}) represent particular
cases of the three-term recurrence relation
\be
\sqrt{2 \lambda_{n+1}}\,f_{n+1} + \sqrt{2 \lambda_{n}}\, f_{n-1} = 2 X_5 f_n 
\label{c}
\ee
for $n=3$ and $n=2$, respectively. Note the similarity between the recurrence relations
(\ref{CRR})  and (\ref{c}).

Finally, from (\ref{A(int)f1F}) it follows that $ A^{\intercal}_5\,f_1 = \sqrt{\lambda_2}
\, f_2$; hence
\be
\sqrt{\lambda_2}\, f_2 = A^{\intercal}_5\,f_1 = ({\sqrt 2}\, X_5 - A_5)\,f_1 = 
{\sqrt 2}\, X_5 f_1 - \sqrt{\lambda_1}\, {\eta}\,\Big(f_0 + \frac{{\sqrt 5}\,c_2}{4}\,f_4\Big)\,,
\label{d}
\ee
where the readily verifiable identity 
\be
A_5\,f_1\,=\,\sqrt{\lambda_1}\,\Big(\cos{\varphi}\,f_0 - \sin{\varphi}
\,f_4\Big)=\sqrt{\lambda_1}\,{\eta}\,\,\Big(f_0 + \frac{{\sqrt 5}\,c_2}{4}\,f_4\Big)
\label{e}
\ee
has been used. Thus, Eq. (\ref{d}) represents a four-term recurrence relation 
\be
\sqrt{2 \lambda_2}\, f_2  + \sqrt{2 \lambda_1}\, {\eta}\,\Big(f_0 + \frac{{\sqrt 5}
\,c_2}{4}\,f_4\Big)= 2\, X_5 f_1\,,
\label{4RR}
\ee
where a linear combination of the eigenvectors $f_0$ and $f_4$ appears on the left 
side, instead of only $f_0$ [\,cf.(\ref{e})]. This difference between (\ref{c}) and 
(\ref{4RR}) is a consequence of the fact that operators $X_5$ and $Y_5=-\,{\rm i}D_5$  
do not satisfy the Heisenberg commutation relation $[x,p] = {\rm i}$ for the standard 
operators $x$ and $p$ in quantum mechanics (see \cite{AAZh} for a more detailed 
discussion of this point).

{\bf 3°}. Finally, the above formulas for the explicit form of the eigenvectors 
$f_k,\,1\leq k \leq 4$, allow us to represent them in the form  $f_k = d^{\,-1}_k
\,{\mathcal P}_k (X_5) f_0$, where ${\mathcal P}_k (X_5)$ is a polynomial in  
matrix $X_5$ of degree $k$. In this way one can obtain an explicit form of the 
discrete analog of the second part of formula (\ref{psi}) associated with 
the continuous case. This can be ascertained as follows.

From (\ref{A(int)f0})--(\ref{f1}) one readily derives that 
\be
f_1\,=\,\Big(\eta\,\sqrt{\lambda_1}\Big)^{-1} A^{\intercal}_5\,f_0 = 
\Big(\eta\,\sqrt{2 \lambda_1}\Big)^{-1} \Big(X_5 + {\mathcal B}^{(s)}\Big)f_0 
= d^{\,-1}_1\,{\mathcal P}_1(X_5) f_0\,,
\label{f1P1}
\ee
where ${\mathcal P}_1(X_5):=X_5 + {\mathcal B}^{(s)}$ and $d_1 = \eta\,\sqrt{2 \lambda_1}$.

From (\ref{A(int)f1}) and (\ref{A(int)f1F}) one similarly derives that 
\be
f_2\,=\,\frac{1}{\sqrt{\lambda_2}}\,A^{\intercal}_5\,f_1 = 
\Big(2 \eta\,\sqrt{\lambda_1 \lambda_2}\Big)^{-1} 
\Big(X_5 - s^{-1}_2 {\mathcal B}^{(a)}\Big) \Big(X_5 + {\mathcal B}^{(s)}\Big)f_0 
= d^{\,-1}_2\,{\mathcal P}_2(X_5) f_0\,,
\label{f2P2}
\ee
where
\be
{\mathcal P}_2(X_5) = \Big(X_5 - s^{-1}_2 {\mathcal B}^{(a)}\Big) \Big(X_5 
+ {\mathcal B}^{(s)}\Big), \quad d_2 = 2 \eta\,\sqrt{\lambda_1 \lambda_2}\,.
\label{P2}
\ee

Also, from (\ref{A(int)f2})--(\ref{f3}) one similarly gets that 
\be
f_3\,=\,\frac{1}{\sqrt{\lambda_3}}\,A^{\intercal}_5\,f_2 = 
d^{\,-1}_3 \Big(X_5 - {\mathcal B}^{(s)}\Big)\Big(X_5 - s^{-1}_2 {\mathcal B}^{(a)}\Big) 
\Big(X_5 + {\mathcal B}^{(s)}\Big)f_0 = d^{\,-1}_3\,{\mathcal P}_3(X_5) f_0\,,
\label{f3P3}
\ee
where
\be
{\mathcal P}_3(X_5) =  \Big(X_5 - {\mathcal B}^{(s)}\Big)
\Big(X_5 - s^{-1}_2 {\mathcal B}^{(a)}\Big) \Big(X_5 + {\mathcal B}^{(s)}\Big), 
\quad d_ 3 = \eta\,\sqrt{2^3 \lambda_1 \lambda_2 \lambda_3}.
\label{P3}
\ee
Finally, from (\ref{A(int)f3}), (\ref{v}), and  (\ref{f4}), it follows that 
\be
f_4\,=\,\frac{1}{\sqrt{\lambda_4}}\,A^{\intercal}_5\,f_3 = 
d^{\,-1}_4\,{\mathcal P}_4(X_5) f_0\,,
\label{f4P4}
\ee
where
\be
{\mathcal P}_4(X_5) =  \Big(X_5 + s^{-1}_2 {\mathcal B}^{(a)}\Big) 
\Big(X_5 - {\mathcal B}^{(s)}\Big) \Big(X_5 - s^{-1}_2 {\mathcal B}^{(a)}\Big) 
\Big(X_5 + {\mathcal B}^{(s)}\Big), \quad 
d_ 4 = 4 \eta\,\sqrt{ \lambda_1 \lambda_2 \lambda_3 \lambda_4}\,.
\label{P4}
\ee

Thus, the discrete analog of the formula $\psi_n(x) = c_n^{-1} H_n(x)\,\psi_0(x),\,\, 
c_n= \sqrt{2^n\, n!}$\,, associated with the continuous case, has the form: 
\be
f_n = d^{\,-1}_n\,{\mathcal P}_n (X_5) f_0, \qquad n=1,2,3,4, 
\label{fng}
\ee
where $d_n = \eta \prod^n_{k=1} (2 \lambda_k)^{1/2}$  and ${\mathcal P}_n (X_5)$  
are the Newtonian basis matrix polynomials in $X_5$ (see, e.g.\,\cite{VZh}--\cite{LVS3} 
and relevant references quoted therein on various applications of the Newtonian basis), 
defined as
\be
{\mathcal P}_0(X_5) = 1 , \qquad {\mathcal P}_n (X_5) = (X_5 - M_{n-1})
\cdots (X_5 - M_1)(X_5 - M_0),  \quad n=1,2,3,4\,,
\label{PnX5}
\ee
with the matrices $M_0=- M_2= -\,{\mathcal B}^{(s)}$ and  $M_1=- M_3= 
s^{-1}_2 {\mathcal B}^{(a)}$ as interpolation nodes at 0,1,2,3. It should be noted 
that by combining (\ref{fng}) with the orthonormality condition of eigenvectors $f_n$, 
it is easy to verify that  the polynomials ${\mathcal P}_n (X_5)$ are orthogonal:
\be
\Big({\mathcal P}_k(X_5) f_0,\,{\mathcal P}_l (X_5)  f_0\Big) = d^{\,2}_k\,\delta_{kl}\,.
\label{PnOrt}
\ee
Simultaneously, it is extremely important to emphasize that the matrix polynomials 
${\mathcal P}_n (X_5)$  do not belong to the class of  hypergeometric-type polynomials 
on a Newtonian basis, which necessarily satisfy the standard three-term recurrence 
relations  \cite{VZh}.

\section{Conclusions}

To summarize, the eigenvalues $\lambda_n$ and eigenvectors $f_n$ of the $5D$ discrete
number operator $\mathcal{N}_5 = A^{\intercal}_5\, A_5$ are evaluated in a systematic 
way. Because the eigenvalues $\lambda_n$ are represented by distinct non-negative numbers, 
the number operator  $\mathcal{N}_5$ has been used to classify eigenvectors of the $5D$ 
discrete Fourier transform $\Phi_5$, thus resolving the ambiguity caused by the well-known 
degeneracy of the eigenvalues of the discrete Fourier transform $\Phi_N$. A procedure for 
{\it sparsealization} the intertwining operators $A_5$ and $A^{\intercal}_5$ has been 
formulated, which made it possible to construct the discrete analog (\ref{fnf0}) of the 
well-known continuous-case formula $\psi_n(x) = \frac{1}{\sqrt{n!}}\, ({\bf {a}^{\dagger}})^n 
\psi_0(x)$. In addition, a discrete analog for the eigenvectors $f_n$ of the continuous-case 
formula $\psi_n(x) = c_n^{-1} H_n(x)\,\psi_0(x),\,\,c_n= \sqrt{2^n\, n!}$, has been established 
in terms of the Newtonian basis polynomials ${\mathcal P}_n (X_5),\,n\in {\mathbb Z_5}$, times 
the lowest eigenvector $f_0$. The methodology developed not only deepens the understanding 
of the discrete Fourier transform but also paves the way for advanced numerical methods in spectral
analysis. 

{\Large\bf Acknowledgments}

I am grateful to Luis Verde-Star and Alexei Zhedanov for their illuminating discussions.



\bb{99}

\bi{McCPar} J.H.McClellan and T.W.Parks, {\it Eigenvalue and eigenvector decomposition 
of the discrete Fourier transform}, IEEE Trans. Audio Electroac., {\bf AU-20}, 66--74, 1972.

\bi{AusTol} L.Auslander and R.Tolimieri, {\it Is computing with the finite Fourier transform 
pure or applied mathematics?}  Bull. Amer. Math. Soc., {\bf 1}, 847--897, 1979.

\bi{DicSte} B.W.Dickinson  and K.Steiglitz, {\it Eigenvectors and functions of the discrete 
Fourier transform}  IEEE Trans. Acoust. Speech, {\bf 30}, 25--31, 1982.

\bi{Mehta} M.L.Mehta, {\it Eigenvalues and eigenvectors of the finite Fourier transform}, 
J. Math. Phys., {\bf 28}, 781--785, 1987.

\bi{Matv} V.B.Matveev, {\it Intertwining relations between the Fourier transform and
discrete  Fourier transform, the related functional identities and beyond}, Inverse Prob., 
{\bf 17}, 633--657, 2001.

\bi{Ata} N.M.Atakishiyev, {\it On $q$-extensions of Mehta's eigenvectors of the finite
Fourier transform}, Int. J. Mod. Phys. A, {\bf 21}, 4993--5006, 2006.

\bi{HJ} R.A.Horn,  C.R.Johnson,  Matrix analysis, Cambridge University Press, 
Cambridge, 2009.

\bi{MesNat} M.K.Atakishiyeva and N.M.Atakishiyev, {\it On the raising and lowering difference 
operators for eigenvectors of the finite Fourier transform}, J. Phys: Conf. Ser., {\bf 597}, 
012012, 2015.

\bi{AA2016} M.K.Atakishiyeva and N.M.Atakishiyev, {\it On algebraic properties of the 
discrete raising and lowering operators, associated with the $N$-dimensional discrete 
Fourier transform}, Adv. Dyn. Syst. Appl., {\bf 11}, 81--92, 2016. 

\bi{4Open} M.K.Atakishiyeva, N.M.Atakishiyev and J.Loreto-Hern{\'a}ndez, {\it More 
on algebraic properties of the discrete Fourier transform raising and lowering operators}, 
4 Open, {\bf 2}, 1--11, 2019.

\bi{AAZh} M.K.Atakishiyeva, N.M.Atakishiyev and A.Zhedanov, {\it An algebraic interpretation 
of the intertwining operators associated with the discrete Fourier transform}, J. Math. Phys., 
{\bf 62}, 101704, 2021.

\bi{Zhe_AW} A.S.Zhedanov, {\it {\lq}{\lq}Hidden symmetry{\rq}{\rq} of Askey-Wilson 
polynomials}, Theoretical and Mathematical Physics {\bf 89}, 1146--1157, 1991.

\bi{Ter_AW} P.Terwilliger, {\it The Universal Askey-Wilson Algebra}, SIGMA {\bf 7}, 069, 2011.

\bi{Tom1} T.H.Koornwinder, {\it The relationship between Zhedanov's algebra AW(3) and
the double affine Hecke algebra in the rank one case}, SIGMA {\bf 3}, 063, 2007.

\bi{Tom2} T.H.Koornwinder, {\it Zhedanov's algebra AW(3) and the double affine Hecke 
algebra in the rank one case.II. The spherical subalgebra}, SIGMA {\bf 4}, 052, 2008.

\bi{Bas_H} P. Baseilhac, S. Tsujimoto, L. Vinet, and A. Zhedanov, {\it The Heun-Askey-Wilson 
Algebra and the Heun Operator of Askey-Wilson Type}, Annales Henri Poincar\'e {\bf 20} 
3091--3112, 2019.

\bi{LL} L.D.Landau, E.M.Lifshitz, Quantum Mechanics (Non-relativistic Theory), Pergamon
Press, Oxford, 1991.

\bi{HarmAnal} M.C.Pereyra and L.A.Ward, Harmonic analysis: from Fourier to wavelets, 
AMS, Providence, Rhode Island, 2012.

\bi{JHW} J.H.Wilkinson, The Algebraic Eigenvalue Problem, Clarendon Press, Oxford, 2004.

\bi{VZh} L.Vinet and A.Zhedanov, {\it Hypergeometric Orthogonal Polynomials
with respect to Newtonian Bases}, SIGMA {\bf 12}, 048, 2011.

\bi{LVS1} L.Verde-Star, {\it Characterization and construction of classical orthogonal 
polynomials using a matrix approach}, Linear Algebra and its Applications., {\bf 438}, 
3635--3648, 2013. 

\bi{LVS2} L.Verde-Star, {\it A unified construction of all the hypergeometric and basic 
hypergeometric families of orthogonal polynomial sequences}, Linear Algebra and its 
Applications., {\bf 627}, 242--274, 2021. 

\bi{LVS3} L.Verde-Star, {\it Linearization and connection coefficients of polynomial sequences: 
A matrix approach}, Linear Algebra and its Applications., {\bf 672}, 195--209, 2023. 

\end{thebibliography}

\end{document}